\documentclass{llncs}
\usepackage{times}
\usepackage{amsmath}
\usepackage{amssymb}
\usepackage{xspace}
\usepackage{xcolor}
\usepackage{paralist}
\usepackage{booktabs}
\usepackage{array,arydshln}
\usepackage{tikz}
\usetikzlibrary{arrows,backgrounds,decorations,decorations.pathmorphing,positioning,fit,automata,shapes,snakes,patterns,plotmarks}
\usepackage{latexsym}
\usepackage{xspace}
\usepackage{url}

\usepackage[linesnumbered, vlined, ruled]{algorithm2e}

\newcolumntype{C}[1]{>{\centering\arraybackslash}p{#1}}
\newcolumntype{L}[1]{>{\raggedright\let\newline\\\arraybackslash\hspace{0pt}}m{#1}}

\DeclareMathAlphabet{\mathpzc}{OT1}{pzc}{m}{it}

\newcommand{\ie}{{i.e.}\xspace}

\newcommand{\eg}{{e.g.}\xspace}

\newcommand{\figref}[1]{Fig.~\ref{fig:#1}}

\newcommand{\tabref}[1]{Table~\ref{tab:#1}}

\newcommand{\ignore}[1]{}

\newcommand{\setof}[1]{\ensuremath{\{#1\}}}

\newcommand{\mtc}[1]{\multicolumn{2}{c}{#1}}

\newcommand{\mypara}[1]{\noindent{\bf #1.}}

\newcommand{\Reals}{\mathbb{R}}

\newcommand{\PosZReals}{\mathbb{R}^{\ge 0}}

\newcommand{\rob}{\rho}
\newcommand{\introb}{[\rob]}
\newcommand{\lowrob}{\ell}
\newcommand{\hirob}{\upsilon}
\newcommand{\domain}{\mathcal{X}}
\newcommand{\Until}{\mathbf{U}}
\newcommand{\vx}{\mathbf{x}}
\newcommand{\timedomain}{\mathcal{T}}

\newcommand{\dinf}{\displaystyle\inf}
\newcommand{\dsup}{\displaystyle\sup}
\newcommand{\dmin}{\displaystyle\min}
\newcommand{\dmax}{\displaystyle\max}

\newcommand{\pvx}[1]{\vx_{[0,{#1}]}}

\newcommand{\trace}{\tau}

\newcommand{\horizon}{T_{H}}

\newcommand{\SynTree}[1]{\mathcal{T}_{#1}}
\newcommand{\scope}{\mathsf{hor}}
\newcommand{\parent}{\mathsf{parent}}
\newcommand{\treeNode}{v}
\newcommand{\rootOf}{\mathsf{root}}
\newcommand{\op}{\mathsf{op}}
\newcommand{\tem}{\mathbf{H}}
\newcommand{\alw}{\Box}
\newcommand{\ev}{\Diamond}

\newcommand{\runningformula}{\alw_{[0,a]}\left(\neg(y>0)\vee\ev_{[b,c]}(x>0)\right)}

\newcommand{\rsi}{$\mathtt{RoSI}$\xspace}
\newcommand{\rsis}{$\mathtt{RoSI}$s\xspace}

\newcommand{\worklist}{\mathsf{worklist}}

\newcommand{\completions}{\mathcal{C}}
\newcommand{\flo}{f_{\mathsf{inf}}}
\newcommand{\fhi}{f_{\mathsf{sup}}}

\newcommand{\vy}{\mathbf{y}}

\newcommand{\vxl}{{\vx_{\mathsf{inf}}}}
\newcommand{\vxu}{{\vx_{\mathsf{sup}}}}
\newcommand{\trivial}{(\vxl,\vxu)}
\newcommand{\bl}[1]{\textcolor{blue}{#1}}
\newcommand{\gr}[1]{\textcolor{black!60!green}{#1}}
\newcommand{\rd}[1]{\textcolor{red}{#1}}

\newcommand{\procName}{\mathsf{updateWorkList}}

\newcommand{\assign}{\mathtt{:=\ }}

\newcommand{\newTime}{t_{i+1}}
\newcommand{\newValue}{\vx_{i+1}}

\newcommand{\last}{\mathsf{last}}

\def\f {\varphi}

\newcommand{\tbAF}{\texttt{avoid\_front}}
\newcommand{\tbAL}{\texttt{avoid\_left}}
\newcommand{\tbAR}{\texttt{avoid\_right}}
\newcommand{\tbHCa}{\texttt{hill\_climb}_1}
\newcommand{\tbHCb}{\texttt{hill\_climb}_2}
\newcommand{\tbHCc}{\texttt{hill\_climb}_3}
\newcommand{\tbF}{\texttt{filter}}
\newcommand{\tbKB}{\texttt{keep\_bump}}
\newcommand{\tbWH}{\texttt{what\_hill}}

\newcommand{\fkb}{\varphi_{\texttt{keep\_bump}}}

\newcommand{\bumpr}{\texttt{bump\_right}}
\newcommand{\bumpl}{\texttt{bump\_left}}

\newcommand{\rt}{\textsuperscript{\textregistered}}

\newcommand{\fover}{\varphi_\mathit{overshoot}}
\newcommand{\ftemp}{\varphi_\mathit{transient}}
\newcommand{\abs}[1]{\left|#1\right|}
\newcolumntype{C}[1]{>{\centering\arraybackslash}p{#1}}
\newcolumntype{L}[1]{>{\raggedright\let\newline\\\arraybackslash\hspace{0pt}}m{#1}}
\tikzstyle{smalltext}=[font=\fontsize{7}{7}\selectfont]
\tikzstyle{normaltext}=[font=\fontsize{8}{8}\selectfont]
\tikzstyle{bigtext}=[font=\fontsize{10}{10}\selectfont]
\tikzstyle{block}=[draw,rectangle,minimum height=2em,minimum width=3em,smalltext,align=center]
\tikzstyle{val}=[mark=square*, mark options={solid,fill=black}, mark size=1pt]
\tikzstyle{valb}=[mark=square*, mark options={solid,fill=blue}, mark size=1pt]
\tikzstyle{valr}=[mark=square*, mark options={solid,fill=red}, mark size=1pt]
\tikzstyle{valg}=[mark=square*, mark options={solid,fill=black!60!green}, mark size=1pt]
\newcommand{\dada}{\text{-}\text{-}}
\newcommand{\edge}{\mathtt{F}}
\newcommand{\head}[1]{\min(#1)}
\newcommand{\tail}[1]{\max(#1)} 
\newcommand{\slidingmax}{\mathsf{SlidingMax}}

\newcommand{\popBack}[1]{#1 \assign #1 - \tail{#1}}
\newcommand{\popFront}[1]{#1 \assign #1 - \head{#1}}
\newcommand{\pushBack}[2]{#1 \assign #1 \cup \setof{#2}}

\numberwithin{equation}{section}

\makeatletter
\g@addto@macro \normalsize {%
 \setlength\abovedisplayskip{2pt plus 0pt minus 2pt}
 \setlength\belowdisplayskip{4pt plus 2pt minus 2pt}
}
\title{Robust Online Monitoring of Signal Temporal Logic\\[.5cm]}

\author{Jyotirmoy V. Deshmukh\inst{1} \and
        Alexandre Donz\'{e}\inst{2} \and
        Shromona Ghosh\inst{2} \\ 
        Xiaoqing Jin \inst{1} \and
        Garvit Juniwal\inst{2} \and
        Sanjit A.  Seshia\inst{2}}
\institute{Toyota Technical Center,
        \email{$\mathtt{firstname.lastname@tema.toyota.com}$} \and 
           University of California Berkeley,
           \email{$\mathtt{\{donze, shromona.ghosh, garvitjuniwal, sseshia\}@eecs.berkeley.edu}$}}

\begin{document}

\maketitle

\begin{abstract}
Signal Temporal Logic (STL) is a formalism used to rigorously specify requirements of cyberphysical
systems (CPS), i.e., systems mixing digital or discrete components in interaction with a continuous
environment or analog components. STL is naturally equipped with a quantitative semantics which can
be used for various purposes: from assessing the robustness of a specification to guiding searches
over the input and parameter space with the goal of falsifying the given property over system
behaviors. Algorithms have been proposed and implemented for \emph{offline} computation of such
quantitative semantics, but only few methods exist for an \emph{online} setting, where one would
want to monitor the satisfaction of a formula during simulation.  In this paper, we formalize a
semantics for robust online monitoring of \emph{partial} traces, i.e., traces for which there might
not be enough data to decide the Boolean satisfaction (and to compute its quantitative counterpart).
We propose an efficient algorithm to compute it and demonstrate its usage on two large scale
real-world case studies coming from the automotive domain and from CPS education in a Massively Open
Online Course (MOOC) setting. We show that savings in computationally expensive simulations far
outweigh any overheads incurred by an online approach.


\end{abstract}

\section{Introduction}\label{sec:intro}
Design engineers for embedded control software typically validate their
designs by inspecting concrete observations of system behavior.  For
instance, in the model-based development (MBD) paradigm, designers have
access to numerical simulation tools to obtain traces from models of
systems.  An important problem is then to be able to efficiently test
whether some logical property $\f$ holds for a given simulation trace. It
is increasingly common \cite{jin2013mining,FainekosSUY12acc,HoxhaAF14arch1,bartocci2014data,kong2014temporal}
to specify such properties using a real-time temporal logic such as Signal
Temporal Logic (STL) \cite{donze2010robust} or Metric Temporal Logic (MTL)
\cite{FainekosP09}.  An {\em offline monitoring} approach involves
performing an {\em a posteriori} analysis on {\em complete} simulation
traces (\ie, traces starting at time $0$, and lasting till a user-specified
time horizon).  Theoretical and practical results for offline monitoring
\cite{FainekosP09,donze2013robustmonitor,donze2010robust,MalerN04} focus on
the efficiency of monitoring as a function of the length of the trace, and
the size of the formula representing the property $\f$.

There are a number of situations where offline monitoring is unsuitable.
Consider the case where the monitor is to be deployed in an actual system
to detect erroneous behavior.  As embedded software is typically resource
constrained, offline monitoring -- which requires storing the entire
observed trace -- is impractical. Also, when a monitor is used in a
simulation-based validation tool, a single simulation may run for several
minutes or even hours.  If we wish to monitor a safety property over the
simulation, a better use of resources is to abort the simulation whenever a
violation is detected.  Such situations demand an {\em online monitoring
algorithm}, which has markedly different requirements.  In particular, a
good online monitoring algorithm must: (1) be able to generate intermediate
estimates of property satisfaction based on {\em partial signals}, (2) use
minimal amount of data storage, and (3) be able to run fast enough in a
real-time setting. 

Most works on online monitoring algorithms for logics such as Linear
Temporal Logic (LTL) or Metric Temporal Logic (MTL) have focussed on the
Boolean satisfaction of properties by partial signals
\cite{OuakOnlineMTL2014,eisner_truncated_2003,nickovic2007amt}. However,
recent work has shown that by assigning quantitative semantics to real-time
logics such as MTL and STL, problems such as bug-finding, parameter
synthesis, and robustness analysis can be solved using powerful
off-the-shelf optimization tools \cite{staliro,breach}.  A robust
satisfaction value is a function mapping a property $\f$ and a trace
$\vx(t)$ to a real number.  A large positive value suggests that $\vx(t)$
easily satisfies $\f$, a positive value close to zero suggests that
$\vx(t)$ is close to violating $\f$, and a negative value indicates a
violation of $\f$. While the recursive definitions of quantitative
semantics naturally define offline monitoring algorithms to compute robust
satisfaction values
\cite{FainekosP09,donze2010robust,donze2013robustmonitor}, there is limited
work on an online monitoring algorithm to do the same
\cite{dokhanchi2014line}.

The main technical and theoretical challenge of online monitoring lies in
the definition of a practical semantics for a temporal logic formula over a
partial signal, i.e., a signal trace with incomplete data which cannot yet
validate or invalidate $\f$. Past work \cite{eisner_truncated_2003} has
identified three views for the satisfaction of a LTL property $\f$ over a
partial trace $\trace$: (1) a {\em weak view} where the truth value of $\f$
over $\trace$ is assigned to {\em true} if there is some suffix of $\trace$
that satisfies $\f$, (2) a {\em strong view} when it is defined to be {\em
false} when some suffix of $\trace$ does not satisfy $\f$ and (3) a {\em
neutral view} when the truth value is defined using a truncated semantics
of LTL restricted to {\em finite} paths. In \cite{OuakOnlineMTL2014}, the
authors extend the truncated semantics to MTL, and in
\cite{dokhanchi2014line}, the authors introduce the notion of a
\emph{predictor}, which works as an oracle to complete the partial trace
and provide an estimated satisfaction value.  However, such a value cannot
be formally trusted in general as long as the data is incomplete. 

We now outline our major contributions in this paper.  In
Section~\ref{sec:rosi}, we present {\em robust interval} semantics for an
STL property $\f$ on a partial trace $\trace$ that unifies the different
semantic views of real-time logics on truncated paths.  Informally, the
robust interval semantics map a trace $\vx(t)$ and an STL property $\f$ to
an interval $(\lowrob,\hirob)$, with the interpretation that for any suffix
$u(t)$, $\lowrob$ is the greatest lower bound on the quantitative semantics
of the trace $\vx(t)$, and $\hirob$ is the corresponding lowest upper
bound. There is a natural correspondence between the interval semantics and
three-valued semantics: (1) the truth value of $\f$ is false according to
the weak view iff $\hirob$ is negative, and true otherwise; (2) the truth
value is true according to the strong view iff $\lowrob$ is positive, and
false otherwise; and (3) a neutral semantics, e.g., based on some
predictor, can be defined when $\lowrob<0<\hirob$, i.e., when there exist
both suffixes that can violate or satisfy $\f$.

In Section~\ref{sec:algo}, we present an efficient online algorithm to
compute the robust interval semantics for bounded horizon formulas. Our
approach is based on the offline algorithm of \cite{donze2013robustmonitor}
extended to work in a fashion similar to the incremental Boolean monitoring
of STL implemented in the tool AMT \cite{nickovic2007amt}. A key feature of
our algorithm is that it imposes minimal runtime overhead with respect to
the offline algorithm, while being able to compute robust satisfaction
intervals on partial traces.  In Section~\ref{sec:untimed}, we present
specialized algorithms to deal with commonly-used unbounded horizon
formulas using only a bounded amount of memory.

Finally, we present an implementation and experimental results on two
large-scale case studies: (i) industrial-scale Simulink models from the
automotive domain in Section~\ref{sec:experiments}, and (ii)  an automatic
grading system used in a massive online education initiative on
CPS~\cite{juniwal-emsoft14}. Since the online algorithm can abort
simulation as soon as the truth value of the property is determined, we see
a consistent 10\%-20\% savings in simulation time (which is typically
several hours) in a majority of experiments, with negligible overhead
($<1$\%). In general, our results indicate that the benefits of our online
monitoring algorithm over the offline approach far outweigh any overheads.


\section{Background}\label{sec:interval_robustness}
\label{sec:prelim}
\mypara{Interval Arithmetic} We now review interval arithmetic.  An
interval $I$ is a convex subset of $\Reals$. A singular interval $[a,a]$
contains exactly one point.  Intervals $(a,a)$, $[a,a)$, $(a,a]$, and
$\emptyset$ denote empty intervals.  We enumerate interval operations below
assuming open intervals.  Similar operations can be defined for closed,
open-closed, and closed-open intervals. 
\begin{equation}
\begin{array}[t]{l@{\hspace{1em}}l}
\begin{array}[t]{llcl}
1. & -I_1   & = & (-b_1,-a_1) \\
2. & c+I_1  & = & (c+a_1, c+b_1) 
\end{array} &
\begin{array}[t]{llcl}
3. & I_1 \oplus I_2 & = & (a_1+a_2,b_1+b_2)  \\
4. & \min(I_1,I_2)  & = & (\min(a_1,a_2), \min(b_1,b_2)) 
\end{array} \\
\multicolumn{2}{l}{%
5.\ I_1 \cap I_2    = \left\{\begin{array}{cl}
                            \emptyset                      & \text{if $\min(b_1,b_2) < \max(a_1,a_2)$} \\
                            (\max(a_1,a_2), \min(b_1,b_2)) & \text{otherwise.}
                            \end{array}
                     \right.
}
\end{array}
\end{equation}

\begin{definition}[Signal]
A {\em time domain} $\timedomain$ is a finite or infinite set of time
instants such that $\timedomain \subseteq \PosZReals$ with $0\in \timedomain$. A {\em signal} $\vx$ is a
function from $\timedomain$ to $\domain$. Given a time domain $\timedomain$, a {\em partial signal}
is any signal defined on a time domain $\timedomain'\subseteq\timedomain$.
\end{definition}

Simulation frameworks typically provide signal values at discrete time
instants, usually this is a by-product of using a numerical technique to
solve the differential equations in the underlying system.  These
discrete-time solutions are assumed to be sampled versions of the actual
signal, which can be reconstructed using some form of interpolation.  In
this paper, we assume constant interpolation to reconstruct the signal
$\vx(t)$, \ie, given a sequence of time-value pairs $(t_0,\vx_0), \ldots,
(t_n,\vx_n)$, for all $t \in [t_0,t_n)$, we define $\vx(t) = \vx_i$ if $t
\in [t_i,t_{i+1})$, and $\vx(t_n) = \vx_n$. Further, let $\timedomain_n
\subseteq \timedomain$ represent the finite subset of time instants at
which the signal values are given.

\mypara{Signal Temporal Logic} We use Signal
Temporal Logic (STL) \cite{donze2010robust} to analyze time-varying
behaviors of signals.  We now present its syntax and semantics.  A
{\em signal predicate} $\mu$ is a formula of the form $f(\vx) > 0$,
where $\vx$ is a variable that takes values from $\domain$, and $f$ is
a function from $\domain$ to $\mathbb{R}$.  For a given $f$, let
$\flo$ denote $\inf_{\vx \in \domain} f(\vx)$, \ie, the {\em greatest
lower bound} of $f$ over $\domain$.  Similarly, let $\fhi = \sup_{\vx
\in \domain} f(\vx)$.  The syntax of an STL formula $\varphi$ is
defined in Eq.~\eqref{eq:stl_syntax}. Note that $\alw$ and  $\ev$ can
be defined in terms of the $\Until$ operator, but we include them for
convenience.
\begin{equation}
\label{eq:stl_syntax}
    \varphi ::=  \mu                             \mid 
                 \neg \varphi                    \mid 
                 \varphi \wedge \varphi          \mid 
                 \alw_{(u,v)} \varphi             \mid 
                 \ev_{(u,v)} \varphi        \mid 
                 \varphi \Until_{(u,v)} \varphi  
\end{equation}

Quantitative semantics for timed-temporal logics have been proposed
for STL in \cite{donze2010robust}; we include the definition below.

\begin{definition}[Robust Satisfaction Value]
The {\em robust satisfaction value} is a function $\rob$ mapping
$\varphi$, the  signal $\vx$, and a time $\tau \in \timedomain$ as
follows:
\begin{equation}
\begin{array}{l@{\hspace{2em}}c@{\hspace{2em}}l}
\rob\left(f(\vx) > 0, \vx, \tau\right) & = & 
f(\vx(\tau)) \\
\rob\left(\neg \varphi, \vx, \tau\right) & = & 
-\rob(\varphi,\vx,\tau) \\
\rob\left(\varphi_1 \wedge \varphi_2, \vx, \tau\right) & = & 
\min\left(\rob(\varphi_1,\vx,\tau),\rob(\varphi_2,\vx,\tau)\right) \\
\rob\left(\alw_I \varphi, \vx, \tau\right) & = & 
\dinf_{\tau'\in\tau+I} \rob(\varphi,\vx,\tau') \\
\rob\left(\ev_I \varphi, \vx, \tau\right) & = & 
\dsup_{\tau'\in\tau+I} \rob(\varphi,\vx,\tau') \\
\rob\left(\varphi \Until_I \psi, \vx, \tau\right) & = & 
\dsup_{\tau_1 \in\tau+I} 
     \min 
         \left( 
                \rob(\psi, \vx, \tau_1), 
                \dinf_{\tau_2 \in (\tau, \tau_1)} \rob(\varphi, \vx, \tau_2)
          \right)
\end{array}
\end{equation}
\end{definition}

Here, the translation from quantitative semantics to the usual Boolean
satisfaction semantics is that a signal $\vx$ satisfies an STL formula
$\varphi$ at a time $\tau$ iff the robust satisfaction value
$\rob(\varphi,\vx,\tau) \ge 0$.


\section{Robust Interval Semantics}\label{sec:rosi}
In what follows, we assume that we wish to monitor the robust
satisfaction value of a signal over a finite time-horizon $\horizon$.
We assume that the signal is obtained by applying piecewise constant
interpolation to a sampled signal defined over time-instants
$\setof{t_0,t_1,\ldots,t_N}$, such that $t_N = \horizon$.  In an
online monitoring context, at any time $t_i$, only the partial signal
over time instants $\setof{t_0, \ldots, t_i}$ is available, and the
rest of the signal becomes available in discrete time increments.  We
define robust satisfaction semantics of STL formulas over such partial
signals using an interval-based semantics. Such a {\em robust
satisfaction interval} (\rsi) includes all possible robust
satisfaction values corresponding to the suffixes of the
partial signal.  In this section, we formalize the recursive
definitions for the robust satisfaction interval of an STL formula
with respect to a partial signal, and in the next section we will
discuss an efficient algorithm to compute and maintain these
intervals.

\begin{definition}[Prefix, Completions]
Let \mbox{$\{t_0$, $\ldots$, $t_i\}$} be a finite set of time instants such
that $t_i \leq \horizon$, and let $\pvx{i}$ be a partial signal over the
time domain $[t_0,t_i]$.  We say that $\pvx{i}$ is a
prefix of a signal $\vx$ if for all $t \leq t_i$, $\vx(t) =
\pvx{i}(t)$.  The {\em set of completions} of a partial signal
$\pvx{i}$ (denoted by $\completions(\pvx{i})$) is defined as the set
$\setof{\vx \mid \text{$\pvx{i}$ is a prefix of $\vx$}}$.
\end{definition}

\begin{definition}[Robust Satisfaction Interval (\rsi)]
\label{def:rsi} The robust satisfaction interval of an STL formula
$\varphi$ on a partial signal $\pvx{i}$ at a time $\tau \in [t_0,t_N]$
is an interval 
$I$ such that:
\[
\inf(I) = \displaystyle \inf_{\vx \in \completions(\pvx{i})} \rob(\varphi,\vx,\tau) 
\qquad\mathrm{and}\qquad 
\sup(I) = \displaystyle \sup_{\vx \in \completions(\pvx{i})} \rob(\varphi,\vx,\tau)
\]
\end{definition}

\begin{definition}
\renewcommand{\arraystretch}{1.4}
\label{def:introb}
We now define a recursive function $\introb$ that maps a given formula
$\varphi$, a partial signal $\pvx{i}$ and a time $\tau \in \timedomain$ to an
interval $\introb(\varphi,\pvx{i},\tau)$.

\begin{equation}\label{eq:introb}
\begin{array}{l@{\hspace{1em}}c@{\hspace{1em}}l} 
\introb\left(f(\pvx{i})>0, \pvx{i}, \tau\right) & = &
\left\{
    \begin{array}{ll}
       [f(\pvx{i}(\tau)),f(\pvx{i}(\tau))] & \tau \in [t_0,t_i] \\
       \,[\flo,\fhi] & \text{otherwise.}
    \end{array}
\right.  \\
\introb\left(\neg \varphi, \pvx{i}, \tau\right) & = &
-\introb(\varphi, \pvx{i}, \tau) \\
\introb\left(\varphi_1 \wedge \varphi_2, \pvx{i}, \tau\right) & = &
\min(\introb(\varphi_1,\pvx{i},\tau), \introb(\varphi_2,\pvx{i},\tau)) \\
\introb\left(\alw_I \varphi, \pvx{i}, \tau\right) & = &
\dinf_{t \in \tau+ I} \left(\introb(\varphi,\pvx{i},\tau)\right) \\
\introb\left(\ev_I \varphi, \pvx{i}, \tau\right) & = &
\dsup_{t \in \tau+ I} \left(\introb(\varphi,\pvx{i},\tau)\right) \\
\introb\left(\varphi_1 \Until_I \varphi_2, \pvx{i}, \tau\right) & = &
\dsup_{\tau_2 \in \tau+ I} 
     \min\left(\begin{array}{l}
               \introb(\varphi_2,\pvx{i},\tau_2), \\ 
               \dinf_{\tau_1 \in (\tau,\tau_2)} \introb(\varphi_1,\pvx{i},\tau_1))
               \end{array}
         \right)
\end{array}
\end{equation}
\end{definition}

The following lemma that can be proved by induction over the structure
of STL formulas shows that the interval obtained by applying the
recursive definition for $\introb$ is indeed the robust satisfaction
interval as defined in Def.~\ref{def:rsi}.

\begin{lemma}
For any STL formula $\varphi$, the function
$\introb(\varphi,\pvx{i},\tau)$ defines the robust satisfaction
interval for the formula $\varphi$ over the signal $\pvx{i}$ at time
$\tau$.
\end{lemma}


\section{Online Algorithm}\label{sec:algo}
Donz{\'{e}} et al.~\cite{donze2013robustmonitor} present an offline
algorithm for monitoring STL formulas over (piecewise) linearly
interpolated signals.  A na\"{i}ve implementation of an online
algorithm is as follows: at time $t_i$, use a modification of the
offline monitoring algorithm to recursively compute the robust
satisfaction intervals as defined by Def.~\ref{def:introb} to the
signal $\pvx{i}$.  We observe that such a procedure does many repeated
computations that can be avoided by maintaining the results of
intermediate computations.  Furthermore, the na\"{i}ve procedure
requires storing the signal values over the entire time horizon, which
makes it memory-intensive.  In this section, we present the main
technical contribution of this paper: {\em an online algorithm that is
memory-efficient and avoids repeated computations}.

As in the offline monitoring algorithm in \cite{donze2013robustmonitor}, an essential ingredient of
the online algorithm is Lemire's running maximum filter algorithm \cite{lemire2006streaming}.  The
problem this algorithm addresses is the following: given a sequence of values $a_1,\ldots,a_n$, find
the maximum (resp. minimum) over all windows of size $w$, \ie, for all $j$, $\max_{i\in[j,j+w)} a_i$
(resp.  $\min_{i\in[j,j+w)} a_i$). We briefly review an extension of Lemire's algorithm over
piecewise-constant signals with variable time steps, given as Algorithm~\ref{algo:lemire}.  The main
observation in Lemire's algorithm is that it is sufficient to maintain a descending
(resp. ascending) monotonic edge (noted $\edge$ in Algorithm~\ref{algo:lemire}) to compute the sliding
maxima (resp. minima), in order to achieve an optimal procedure (measured in terms of the number of
comparisons between elements).

\begin{algorithm}[t]
\caption{SlidingMax($(t_0,\vx_0),\ldots,(t_N,\vx_N)$)}
\label{algo:lemire}
\SetVlineSkip{0em}
\DontPrintSemicolon
\KwIn{Window: $[a,b]$} 
\KwOut{Sliding maximum $\vy(t)$ over times in $[t_0,t_N]$}
$\edge$ $\assign$ $\setof{0}$ \tcp{$\edge\mathtt{~is~the~set~of~times~representing~the~monotonic~edge}$}
$i$ $\assign$ $0$ ; $s,t$ $\assign$ $t_0 - b$ \;
\While{$t + a < t_N$}{
   \lIf{$\edge \neq \emptyset$}{
           $t$ $\assign$ $\min(t_{\head{\edge}}-a, t_{i+1}-b)$ 
   } \lElse {
           $t$ $\assign$ $t_{i+1} - b$ 
   }
   \eIf{$t = t_{i+1} - b$} {
      \While{$\vx_{i+1} \ge \vx_{\tail{\edge}} \ \wedge\ \edge \neq \emptyset$}{
         $\popBack{\edge}$  
      }
      $\pushBack{\edge}{i+1}$, $i$ $\assign$ $i+1$ 
   } (\tcp*[h]{$\mathtt{Slide~window~to~the~right}$}) {
       \lIf{$s > t_0$} {
            $\vy(s)$ $\assign$ $\vx_{\head{\edge}}$  
       } \lElse {
            $\vy(t_0)$ $\assign$ $\vx_{\head{\edge}}$
       }
       $\popFront{\edge}$, $s$ $\assign$ $t$ \;
   }
}
\end{algorithm}


We first focus on the fragment of STL where each temporal operator is
bounded by a time-interval $I$ such that $\sup(I)$ is finite.  The
procedure for online monitoring is an algorithm that maintains in
memory the syntax tree of the formula $\varphi$ to be monitored,
augmented with some book-keeping information. First, we formalize some
notation. For a given formula $\varphi$, let $\SynTree{\varphi}$
represent the syntax tree of $\varphi$, and let
$\rootOf(\SynTree{\varphi})$ denote the root of the tree.  Each node
in the syntax tree (other than a leaf node) corresponds to an STL
operator $\neg, \vee, \wedge, \alw_I$ or $\ev_I$.\footnote{We omit the
case of $\Until_I$ here for lack of space, although the rewriting
approach of \cite{donze2013robustmonitor} can also be adapted and was
implemented in our tool.} We will use $\tem_I$ to denote any temporal
operator bounded by interval $I$.  For a given node $\treeNode$, let
$\op(\treeNode)$ denote the operator for that node.  For any node
$\treeNode$ in $\SynTree{\varphi}$ (except the root node), let
$\parent(\treeNode)$ denote the unique parent of $\treeNode$. 

\begin{figure}[t]
\centering
\begin{tikzpicture}[>=stealth']
\tikzstyle{smalltext}=[font=\fontsize{8}{8}\selectfont]
\tikzstyle{block}=[draw,rectangle,minimum height=1.8em,minimum
width=3em,smalltext,align=center,inner sep=1pt]


\node[block, label=below:{\scriptsize $[0]$}] (root2) 
    {$\Box_{[0,a]}$};
\node[block,node distance=15mm,right of=root2,
     label=below:{\scriptsize $[0,a]$}] (n02)
    {$\bigvee$};
\node[coordinate,node distance=15mm,right of=n02] (c2) {};
\node[block,node distance=6mm,above of=c2,
      label=below:{\scriptsize $[0,a]$}] (n12) 
    {$\neg$};
\node[block,node distance=6mm,below of=c2,
      label=below:{\scriptsize $[0,a]$}] (n22)
    {$\Diamond_{[b,c]}$};
\node[block,node distance=20mm,right of=n12,
      label=below:{\scriptsize $[0,a]$}] (n32)
    {$y > 0$};
\node[block,node distance=20mm,right of=n22,
      label=below:{\scriptsize $[b,a\!+\!c]$}] (n42)
    {$x>0$};
\draw[->] (root2) to (n02);
\draw[->] (n02) to (n12);
\draw[->] (n02) to (n22);
\draw[->] (n12) to (n32);
\draw[->] (n22) to (n42);

\end{tikzpicture}
\caption{Syntax tree $\SynTree{\varphi}$ for $\varphi$ (given in
\eqref{eq:runf}) with each node $\treeNode$ annotated with
$\scope(\treeNode)$. \label{fig:timehorizons}}
\end{figure}


Algorithm~\ref{algo:online_monitoring} does the online \rsi
computation.  Like the offline algorithm, it is a dynamic programming
algorithm operating on the syntax tree of the given STL formula, \ie,
computation of the \rsi of a formula combines the \rsis for its
constituent sub-formulas in a bottom-up fashion.  As computing the
\rsi at a node $\treeNode$ requires the \rsis at the child-nodes, this
computation has to be delayed till the \rsis at the children of
$\treeNode$ in a certain time-interval are available.  We call this
time-interval the {\em time horizon} of $\treeNode$ (denoted
$\scope(\treeNode)$), and define it recursively in
Eq.~\eqref{eq:scope}.
\begin{equation}
\label{eq:scope}
\scope(\treeNode) = 
\left\{ 
\begin{array}{ll}
[0] & \text{if $\treeNode$ = $\rootOf(\SynTree{\varphi})$} \\
I \oplus \scope(\parent(\treeNode)) & 
      \text{if $\treeNode \neq \rootOf(\SynTree{\varphi})$ and $\op(\parent(\treeNode))$ = $\tem_I$} \\
     \scope(\parent(\treeNode)) &
        \text{otherwise.}
    \end{array}
    \right.
\end{equation}
We illustrate the working of the algorithm using a small
example then give a brief sketch of the various steps in the
algorithm.

\begin{example}
Consider formula \eqref{eq:runf}. We show $\SynTree{\varphi}$ and
$\scope(\treeNode)$ for each node $\treeNode$ in $\SynTree{\varphi}$
in \figref{timehorizons}. In rest of the paper, we use $\varphi$ as a
running example\footnote{We remark that $\varphi$ is equivalent to
$\alw_{[0,a]} \left((y > 0) \implies \ev_{[b,c]}(x>0)\right)$, which
is a common formula used to express a timed causal relation between
two signals.}.
\begin{equation}
\varphi \triangleq \runningformula 
\label{eq:runf} 
\end{equation}
\end{example}

The algorithm augments each node $\treeNode$ of $\SynTree{\varphi}$
with a double-ended queue, that we denote $\worklist[\treeNode]$. Let
$\psi$ be the subformula denoted by the tree rooted at $\treeNode$.
For the partial signal $\pvx{i}$, the algorithm maintains in
$\worklist[\treeNode]$, the \rsi $\introb(\psi, \pvx{i}, t)$ for each
$t \in \scope(\treeNode) \cap [t_0,t_i]$. We denote by
$\worklist[\treeNode](t)$ the entry corresponding to time $t$ in
$\worklist[\treeNode]$.  When a new data-point $\vx_{i+1}$
corresponding to the time $t_{i+1}$ is available, the monitoring
procedure updates each $\introb(\psi,\pvx{i},t)$ in
$\worklist[\treeNode]$ to $\introb(\psi,\pvx{i+1},t)$.

In Fig.~\ref{fig:worklists}, we give an example of a run of the
algorithm. We assume that the algorithm starts in a state where it has
processed the partial signal $\pvx{2}$, and show the effect of
receiving data at time-points $t_3$, $t_4$ and $t_5$.  The figure
shows the states of the worklists at each node of $\SynTree{\varphi}$
at these times when monitoring the STL formula $\varphi$ presented in
Eq.~\eqref{eq:runf}. Each row in the table adjacent to a node shows
the state of the worklist after the algorithm processes the value at
the time indicated in the first column.

The first row of the table shows the snapshot of the worklists at time
$t_2$.  Observe that in the worklists for the subformula $y>0$, $\neg
y>0$, because $a<b$, the data required to compute the \rsi at $t_0$,
$t_1$ and the time $a$, is available, and hence each of the \rsis is
singular.  On the other hand, for the subformula $x>0$, the time
horizon is $[b,a+c]$, and no signal value is available at any time in
this interval.  Thus, at time $t_2$, all elements of
$\worklist[\treeNode_{x>0}]$ are $(\vxl,\vxu)$ corresponding to the
greatest lower bound and lowest upper bound on $x$.

To compute the values of $\ev_{[b,c]}(x>0)$ at any time $t$, we take
the supremum over values from times $t+b$ to $t+c$. As the time
horizon for the node corresponding to $\ev_{[b,c]} (x>0)$ is $[0,a]$,
$t$ ranges over $[0,a]$. In other words, we wish to perform the
sliding maximum over the interval $[0+b, a+c]$, with a window of
length $c-b$.  We can use the algorithm for computing the sliding
window maximum as discussed earlier in this section.  One caveat is
that we need to store separate monotonic edges for the upper and lower
bounds of the \rsis. The algorithm then proceeds upward on the syntax
tree, only updating the worklist of a node only when there is an
update to the worklists of its children.

The second row in each table is the effect of obtaining a new time
point (at time $t_3$) for both signals. Note that this does not affect
$\worklist[\treeNode_{y>0}]$ or $\worklist[\treeNode_{\neg y>0}]$, as
all \rsis are already singular, but does update the \rsi values for
the node $\treeNode_{x>0}$.  The algorithm then invokes
Alg.~\ref{algo:lemire} on $\worklist[\treeNode_{x>0}]$ to update
$\worklist[\treeNode_{\ev_{[b,c]}(x>0)}]$.  Note that in the
invocation on the second row (corresponding to time $t_3$), there is
an additional value in the worklist, at time $t_3$. This leads
Alg.~\ref{algo:lemire} to produce a new value of
$\slidingmax\left(\worklist[\treeNode_{\vx>0}],[b,c]\right)(t_3-b)$,
which is then inserted in $\worklist[\treeNode_{\ev_{[b,c]} x>0}]$.
This leads to additional points appearing in worklists at the
ancestors of this node.

Finally, we remark that the run of this algorithm shows that at time
$t_4$, the \rsi for the formula $\varphi$ is $[-2,-2]$, which yields a negative upper bound, 
showing that the formula is not satisfied irrespective of the suffixes of
$x$ and $y$. In other words, the satisfaction of $\varphi$ is
known before we have all the data required by $\scope(\varphi)$.

\renewcommand*{\arraystretch}{1.3}

\begin{figure*}[!t]
\centering
\begin{tikzpicture}[>=stealth', transform shape, scale=.8]

\draw[densely dotted,very thin,color=gray,step=0.25cm] (0.5,3) grid (7.0,5.75);
\draw[densely dotted,very thin,color=gray,step=0.25cm] (7.5,3) grid (14.0,5.75);

\node[coordinate] (root) at (0,2.5) {};

\draw[->] (1,3.25) -- (1,5.75);
\node[bigtext, anchor=east] at (0.75,5.75) {$x$};

\draw[->] (1,4.5) -- (7,4.5);
\node[bigtext, anchor=north] at (7,4.5) {$t$};

\foreach \y in {0,...,4}
    \draw (0.95,3.5+0.5*\y) -- (1.05,3.5+0.5*\y) {};

\foreach \y in {0,...,4}
    \pgfmathsetmacro\result{\y-2}
    \node[smalltext, fill=white, inner sep=1pt,  anchor=east] at (0.95,3.5+0.5*\y) 
    {\pgfmathprintnumber{\result}};

\draw[draw=brown,densely dotted,thick] (6.5,4) -- (7,4);
\draw[densely dotted,draw=red,thick] (6.5,5.5) -- plot[valr] (6.5,4);
\draw[draw=red,thick] (5.75,5.5) -- (6.5,5.5);
\draw[densely dotted,draw=black!60!green,thick] (5.75,3.5)  -- plot[valg] (5.75,5.5);
\draw[draw=black!60!green,thick] (4.25,3.5) -- (5.75,3.5);
\draw[densely dotted,thick,draw=blue] (4.25,4) -- plot[valb] (4.25,3.5);
\draw[draw=blue,thick]      (3,4)       -- (4.25,4);
\draw[densely dotted,thick] (3,5.5)     -- plot[val] (3,4);
\draw[thick]                (2,5.5)     -- (3,5.5);
\draw[densely dotted,thick] (2,5)       -- plot[val] (2,5.5);
\draw[thick] plot[val]      (1,5)       -- (2,5);

\draw[fill=black] (1,4.5) circle (0.4mm); 
\node[normaltext,anchor=north west] at (1,4.5) {$t_0$};

\draw[fill=blue,draw=blue] (1.75,4.5) circle (0.4mm); 
\draw[->,draw=blue] (1.75,3.75) -- (1.75,4.45);
\node[normaltext,anchor=north] (l1) at (1.75,3.9) {\bl{${t_3-b}$}};

\draw[fill=black] (2,4.5) circle (0.4mm); 
\node[normaltext, anchor=north] at (2,4.5) {$t_1$};

\draw[fill=black!60!green,draw=black!60!green] (2.25,4.5) circle (0.4mm);
\draw[->,draw=black!60!green] (2.25,3.5) -- (2.25,4.45);
\node[normaltext,anchor=north] (l1) at (2.25,3.65) {\gr{$t_4-c$}};

\draw[fill=white,draw=black] (2.5,4.5) circle (0.4mm);
\node[normaltext, anchor=south] at (2.5,4.5) {$a$};

\draw[fill=black] (3,4.5) circle (0.4mm);
\node[normaltext,anchor=north west] at (3,4.5) {$t_2$};

\draw[fill=white,draw=black] (3.5,4.5) circle (0.4mm);
\node[normaltext, anchor=south] at (3.5,4.5) {$b$};

\draw[fill=blue,draw=blue] (4.25,4.5) circle (0.4mm);
\node[normaltext, anchor=north east] at (4.25,4.5) {\bl{$t_3$}};

\draw[fill=white,draw=black] (4.5,4.5) circle (0.4mm);
\node[normaltext, anchor=south] at (4.5,4.5) {$c$};

\draw[fill=black!60!green,draw=black!60!green] (5.75,4.5) circle (0.4mm);
\node[normaltext, anchor=north east] at (5.75,4.5)  {\gr{$t_4$}};

\draw[fill=white,draw=black] (6,4.5) circle (0.4mm);
\node[normaltext, anchor=240] at (6,4.5) {$\mathbf{a+c}$};

\draw[fill=red,draw=red] (6.5,4.5) circle (0.4mm);
\node[normaltext, anchor=north east] at (6.5,4.5) {\rd{$t_5$}};

\draw[->] (8,3.25) -- (8,5.75);
\node[bigtext, anchor=east] at (7.75,5.75) {$y$};

\draw[->] (8,4.5) -- (14,4.5);
\node[bigtext, anchor=north] at (14,4.4) {$t$};

\foreach \y in {0,...,4}
    \draw(7.95,3.5+0.5*\y) -- (8.05,3.5+0.5*\y) {};

\foreach \y in {0,...,4}
    \pgfmathsetmacro\result{\y-2}
    \node[smalltext, fill=white, inner sep=1pt, anchor=east] at (7.95,3.5+0.5*\y) 
    {\pgfmathprintnumber{\result}};

\draw[draw=brown,densely dotted,thick] (13.5,5) -- (14,5);
\draw[draw=red,thick]                 (12.75,5) -- plot[valr] (13.5,5);
\draw[draw=black!60!green,thick]      (11.25,5) -- plot[valg] (12.75,5);
\draw[draw=blue,thick,densely dotted] (11.25,4) -- plot[valb] (11.25,5);
\draw[draw=blue,thick]                (10,4) -- (11.25,4);
\draw[thick,densely dotted]           (10,5.5) -- plot[val] (10,4);
\draw[thick]                          (9,5.5) -- (10,5.5);
\draw[thick,densely dotted]           (9,4)  -- plot[val] (9,5.5);
\draw[thick] plot[val]                (8,4)  -- (9,4);        

\draw[fill=black] (8,4.5) circle (0.4mm); 
\node[normaltext,anchor=north west] at (8,4.5) {$t_0$};

\draw[fill=blue,draw=blue] (8.75,4.5) circle (0.4mm); 
\draw[->,draw=blue] (8.75,5.75) -- (8.75,4.55);
\node[normaltext,anchor=south] (l1) at (8.75,5.7) {\bl{${t_3-b}$}};

\draw[fill=black] (9,4.5) circle (0.4mm); 
\node[normaltext, anchor=60] at (9,4.5) {$t_1$};

\draw[fill=black!60!green,draw=black!60!green] (9.25,4.5) circle (0.4mm);
\draw[->,draw=black!60!green] (9.25,3.5) -- (9.25,4.45);
\node[normaltext,anchor=north] (l1) at (9.25,3.65) {\gr{$t_4-c$}};

\draw[draw=black,fill=white] (9.5,4.5) circle (0.4mm);
\node[normaltext, anchor=south] at (9.5,4.5) {$a$};

\draw[fill=black] (10,4.5) circle (0.4mm);
\node[normaltext,anchor=north west] at (10,4.5) {$t_2$};

\draw[draw=black,fill=white] (10.5,4.5) circle (0.4mm);
\node[normaltext, anchor=south] at (10.5,4.5) {$b$};

\draw[fill=blue,draw=blue] (11.25,4.5) circle (0.4mm);
\node[normaltext, anchor=north east] at (11.25,4.5) {\bl{$t_3$}};

\draw[draw=black,fill=white] (11.5,4.5) circle (0.4mm);
\node[normaltext, anchor=south] at (11.5,4.5) {$c$};

\draw[fill=black!60!green,draw=black!60!green] (12.75,4.5) circle (0.4mm);
\node[normaltext, anchor=north east] at (12.75,4.5)  {\gr{$t_4$}};

\draw[draw=black,fill=white] (13,4.5) circle (0.4mm);
\node[normaltext, anchor=240] at (13,4.5) {$\mathbf{a+c}$};

\draw[fill=red,draw=red] (13.5,4.5) circle (0.4mm);
\node[normaltext, anchor=north] at (13.5,4.5) {\rd{$t_5$}};

\end{tikzpicture}
\caption{These plots show the signals $x(t)$ and $y(t)$.  Each signal
begins at time $t_0 = 0$, and we consider three partial signals:
$\pvx{3}$ (black + blue), and $\pvx{4}$ ($\pvx{3}$ + green), and
$\pvx{5}$ ($\pvx{4}$ + red). \label{fig:plots}}
\end{figure*}
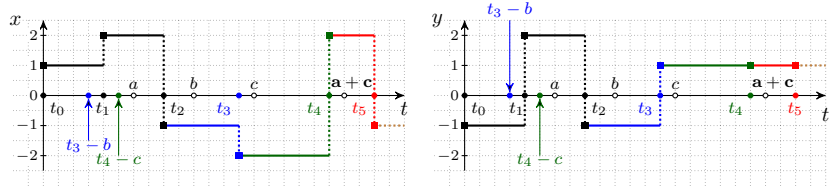

\begin{figure*}[t!]
\centering
\begin{tikzpicture}[>=stealth', transform shape, scale=.85]

\node[block] (root2) at (4,-1) {$\Box_{[0,a]}$};
\node[coordinate] (c) at (4,-1) {};
\node[block,node distance=15mm,below of=c] (n02) {$\bigvee$};
\node[coordinate,node distance=20mm,below of=n02] (c2) {};
\node[block,node distance=10mm,right of=c2] (n12) {$\neg$};
\node[block,node distance=10mm,left of=c2] (n22) {$\Diamond_{[b,c]}$}; 
\node[block,node distance=20mm,below of=n12] (n32) {$y>0$}; 
\node[block,node distance=20mm,below of=n22] (n42) {$x>0$};
\draw[->] (root2) to (n02);
\draw[->] (n02) to (n12);
\draw[->] (n02) to (n22);
\draw[->] (n12) to (n32);
\draw[->] (n22) to (n42);

\newcommand{\minusone}{\text{-1}}
\newcommand{\minustwo}{\text{-2}}

\node[smalltext,anchor=west] (n32t) at (5.5,-6.7) 
{   
    $\begin{array}{l:ccc}
      \hline
                &  t_0 = 0  & t_1  & a   \\
      \hline 
      t_2     &   [\minusone,\minusone] &  [2,2]   &  [2,2] \\
      \hdashline[0.4pt/1pt]
      t_3     &   [\minusone,\minusone] &  [2,2]   &  [2,2] \\
      \hdashline[0.4pt/1pt]
      t_4     &   [\minusone,\minusone] &  [2,2]   &  [2,2] \\
      \hdashline[0.4pt/1pt]
      t_5     &   [\minusone,\minusone] &  [2,2]   &  [2,2] \\
      \hline
    \end{array}$
};

\node[smalltext,anchor=east] (n42t) at (2.5,-6.7) 
    {
      $\begin{array}{l:cccc}
      \hline
              &  b           & t_3                   & t_4   & a\!+\!c \\
      \hline 
      t_2     &  \trivial             & \dada                   & \dada   & \trivial \\
      \hdashline[0.4pt/1pt]
      t_3     & \bl{[\minusone,\minusone]} & \bl{[\minustwo,\minustwo]} & \dada   & \trivial \\
      \hdashline[0.4pt/1pt]
      t_4     & [\minusone,\minusone] & [\minustwo,\minustwo] & \gr{[2,2]} & \trivial \\
      \hdashline[0.4pt/1pt]
      t_5     & [\minusone,\minusone] & [\minustwo,\minustwo] & [2,2] & \rd{[2,2]} \\
      \hline
      \end{array}$
    };

\node[smalltext,anchor=west] (n12t) at (5.5,-4.5)
    {$\begin{array}{l:ccc}
      \hline
               & t_0 = 0    & t_1                   & a \\
      \hline 
      t_2      & [1,1]      & [\minustwo,\minustwo] & [\minustwo,\minustwo] \\
      \hdashline[0.4pt/1pt]                
      t_3      & [1,1]      & [\minustwo,\minustwo] & [\minustwo,\minustwo] \\
      \hdashline[0.4pt/1pt]                       
      t_4      & [1,1]      & [\minustwo,\minustwo] & [\minustwo,\minustwo] \\
      \hdashline[0.4pt/1pt]                       
      t_5      & [1,1]      & [\minustwo,\minustwo] & [\minustwo,\minustwo] \\
      \hline
      \end{array}$
      };

\node[smalltext,anchor=east] (n22t) at (2.5,-4.5)
    {
     $\begin{array}{l:cccc}
      \hline
                 &  t_0 = 0                    & t_3\text{-}b                  & t_4\text{-}c   & a \\
      \hline 
      t_2        & \trivial                    &  \dada                          & \dada            & \trivial \\
      \hdashline[0.4pt/1pt]
      t_3        & \bl{[\minusone,\vxu]}       & \bl{[\minustwo,\vxu]}         & \dada            & \trivial \\
      \hdashline[0.4pt/1pt]
      t_4        & \gr{[\minusone,\minusone]}  & \gr{[\minustwo,\minustwo]}    & \gr{[2,2]}     & \gr{[2,\vxu]} \\
      \hdashline[0.4pt/1pt] 
      t_5        & [\minusone,\minusone]       & [\minustwo,\minustwo]         & [2,2]          & \rd{[2,2]} \\
      \hline
      \end{array}$
    };
      
\node[smalltext,anchor=west] (n02t) at (5.0,-2.4)
    {$\begin{array}{l:ccccc}
      \hline
                &  t_0 = 0          & t_3\text{-}b     & t_1          & t_4\text{-}c  & a  \\
      \hline
      t_2       &  [1,\vxu]         & \dada              & [-2,\vxu]    & \dada           & \trivial    \\
      \hdashline[0.4pt/1pt]
      t_3       &  \bl{[1,\vxu]}    & \bl{[1,\vxu]}    & [-2,\vxu]    & \dada           & \trivial    \\
      \hdashline[0.4pt/1pt]
      t_4       &  \gr{[1,1]}       & \gr{[1,1]}       & \gr{[\minustwo,\minustwo]} & \gr{[2,2]}   & \gr{[2,\vxu]} \\
      \hdashline[0.4pt/1pt]
      t_5       & [1,1]             & [1,1]            & [\minustwo,\minustwo]      & [2,2]        & \rd{[2,2]} \\
      \hline
      \end{array}$
    };

\node[smalltext,anchor=east] (root2t) at (3.0,-1.4)
    {$\begin{array}{l:c}
      \hline
                &  t_0 = 0 \\
      \hline
      t_2       & \trivial  \\
      \hdashline[0.4pt/1pt] 
      t_3       & \trivial  \\
      \hdashline[0.4pt/1pt]
      t_4       & {\bf \gr{[-2,-2]}} \\
      \hdashline[0.4pt/1pt]
      t_5       & \rd{[-2,-2]} \\
      \hline
      \end{array}$
    };

\end{tikzpicture}

\caption{We show a snapshot of the $\worklist[\treeNode]$ maintained by the
algorithm for four different (incremental) partial traces of the signals
$x(t)$ and $y(t)$.  Each row indicates the state of $\worklist[\treeNode]$
at the time indicated in the first column. An entry marked \text{-}\text{-}
indicates that the corresponding element did not exist in
$\worklist[\treeNode]$ at that time. Each colored entry indicates that the
entry was affected by availability of a signal fragment of the
corresponding color.  \label{fig:worklists}} 
\end{figure*}
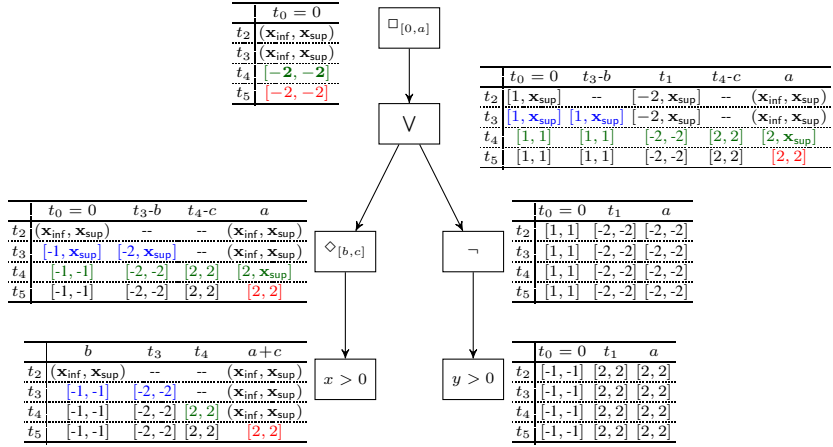



\begin{algorithm}[t]
\SetKwSwitch{Switch}{Case}{Other}{switch}{do}{case}{otherwise}{end}
\DontPrintSemicolon
\tcp{$\treeNode_\psi$ is a node in the syntax tree, $(\newTime,\newValue)$ is a new signal time-point}
\Switch{$\psi$}{
    \Case{$f(\vx) > 0$}{
          \If{$\newTime \in \scope(\treeNode_{\psi})$}{
              $\worklist[\treeNode_{\psi}](\newTime)$ $\assign$ 
              $[f(\newValue),f(\newValue)]$ \nllabel{algoline:l1}
          }
    }
    \Case{$\neg \varphi$}{ 
          $\procName$($\treeNode_{\varphi}$, $\newTime$ ,$\newValue$) 
          \nllabel{algoline:l2} \;
          $\worklist[\treeNode_{\psi}]$ $\assign$ 
            $-\worklist[\treeNode_{\varphi}]$ \nllabel{algoline:l3} 
     } \nllabel{algoline:neg}
     \Case{$\varphi_1 \wedge \varphi_2$}{
          $\procName$($\treeNode_{\varphi_1}$, $\newTime$, $\newValue$)
          \nllabel{algoline:l31} \;
          $\procName$($\treeNode_{\varphi_2}$, $\newTime$, $\newValue$) 
          \nllabel{algoline:l32} \; 
          $\worklist[\treeNode_{\psi}]$ $\assign$
          $\min(\worklist[\treeNode_{\varphi_1}],\worklist[\treeNode_{\varphi_2}])$
          \nllabel{algoline:min}
    }
    \Case{$\alw_I \varphi$}{
          $\procName$($\treeNode_{\varphi}$, $\newTime$ , $\newValue$) \;
          $\worklist[\treeNode_{\psi}]$ $\assign$ $\slidingmax(\worklist[\treeNode_\varphi], I)$ 
    }
}
\caption{$\procName$($\treeNode_\psi$, $\newTime$, $\newValue$)} 
\label{algo:online_monitoring}
\end{algorithm}

Alg.~\ref{algo:online_monitoring} is essentially a procedure that
recursively visits each node in the syntax tree $\SynTree{\varphi}$ of
the STL formula $\varphi$ that we wish to monitor.
Line~\ref{algoline:l1} corresponds to the base case of the recursion,
\ie when the algorithm visits a leaf of $\SynTree{\varphi}$ or an
atomic predicates of the form $f(\vx) > 0$.  Here, the algorithm
inserts the pair $(\newTime,\newValue)$ in
$\worklist[\treeNode_{f(\vx)>0}]$ if $\newTime$ lies inside
$\scope(\treeNode_{f(\vx)>0})$. In other words, it only tracks a value
if it is useful for the computing the robust satisfaction interval of
some ancestor node.

For a node corresponding to a Boolean operation, the algorithm first
updates the worklists at the children, and then uses them to update
the worklist at the node. If the current node represents $\neg
\varphi$ (Line~\ref{algoline:neg}), the algorithm flips the sign of
each entry in $\worklist[\treeNode_\varphi]$; this operation is
denoted as $-\worklist[\treeNode_\varphi]$.  Consider the case where
the current node $\treeNode_{\psi}$ is a conjunction $\varphi_1 \wedge
\varphi_2$.  The sequence of upper bounds and the sequence of lower
bounds of the entries in $\worklist[\treeNode_{\varphi_1}]$ and
$\worklist[\treeNode_{\varphi_1}]$ can be each thought of as a
piecewise-constant signal (likewise for
$\worklist[\treeNode_{\varphi_2}$). In Line~\ref{algoline:min}, the
algorithm computes a pointwise-minimum over piecewise-constant signals
representing the upper and lower bounds of the \rsis of its arguments.
Note that if for $i=1,2$, if $\worklist[\treeNode_{\varphi_i}]$ has
$N_i$ entries, then the pointwise-$\min$ would have to be performed at
most $N_1 + N_2$ distinct time-points. Thus,
$\worklist[\treeNode_{\varphi_1 \wedge \varphi_2}]$ has at most $N_1 +
N_2$ entries. A similar phenomenon can be seen in
Fig.~\ref{fig:worklists}, where computing a $\max$ over the worklists
of $\treeNode_{\ev_{[b,c](x>0)}}$ and $\treeNode_{\neg (y>0)}$ leads
to an increase in the number of entries in the worklist of the
disjunction.

For nodes corresponding to temporal operators, \eg, $\ev_I \varphi$,
the algorithm first updates $\worklist[\treeNode_\varphi]$. It then
applies Alg.~\ref{algo:lemire} to compute the sliding maximum over
$\worklist[\treeNode_\varphi]$. Note that if
$\worklist[\treeNode_\varphi]$ contains $N$ entries, so does
$\worklist[\treeNode_{\ev_I\varphi}]$.

A further optimization can be implemented on top of this basic scheme.
For a node $\treeNode$ corresponding to the subformula $\tem_I
\varphi$, the first few entries of $\worklist[\treeNode]$ (say up to
time $u$) could become singular intervals once the required \rsis for
$\worklist[\treeNode_\varphi]$ are available. The optimization is to
only compute $\slidingmax$ over $\worklist[\treeNode_\varphi]$
starting from $u + \inf(I)$. We omit the pseudo-code for brevity.


\section{Monitoring untimed formulas}\label{sec:untimed}
\newcommand{\subf}{\mathsf{sub}}
\newcommand{\psione}{\varphi}
\newcommand{\psitwo}{\psi}
\newcommand{\WLone}[1]{\mathsf{W_\psione}(#1)}
\newcommand{\WLtwo}[1]{\mathsf{W_\psitwo}(#1)}
\newcommand{\Summary}{\mathsf{S}}
\newcommand{\Sone}{\mathsf{S}_\psione}
\newcommand{\Stwo}{\mathsf{S}_\psitwo}

If the STL formula being monitored has untimed (\ie infinite-horizon)
temporal operators, a direct application of
Alg.~\ref{algo:online_monitoring} requires every node in the sub-tree
rooted at the untimed operator to have an unbounded time horizon. In
other words, for all such nodes, the algorithm would have to keep
track of every value over arbitrarily long intervals.  For certain
untimed operators and the combinations thereof, we show that we can
monitor the formulas using only a bounded amount of information. 

First, we introduce some equivalences over intervals $a,b,c$ that we use in
the theorem and the proof to follow:
\begin{eqnarray}
\min(\max(a,b),\max(a,c)) & = & \max(a, \min(b,c))           \label{eq:equivone} \\
\min(a,\max(b,c))         & = & \max( \min(a,b), \min(a,c) ) \label{eq:equivtwo} \\
\max(\max(a,b), c)        & = & \max(a,b,c)                  \label{eq:equivthree} \\
\min(\max(a,b),a)         & = & a                            \label{eq:equivfour}
\end{eqnarray}

\begin{theorem}
\label{thm:finiteness}
For each of the following formulae, where $\varphi$ and $\psi$ are atomic
predicates of the form $f(\vx)\!>\!0$, we can monitor interval robustness in an
online fashion using constant memory: (1) $\alw \varphi$, $\ev \varphi$,
(2) $\varphi \Until \psi$, (3) $\alw(\varphi \vee \ev \psi)$, $\ev(\varphi \wedge
\alw \psi)$, (4) $\alw \ev \varphi$, $\ev \alw \varphi$, and (5)
$\ev(\varphi \wedge \ev \psi)$, $\alw(\varphi \vee \alw \psi)$.
\end{theorem}

\begin{proof} 
In what follows, we use the following short-hand notation:
\begin{equation}
\label{eq:pq}
\begin{array}{lll@{\hspace{2em}}lll}
p_i & \equiv & \introb(f(\vx)\!>\!0, \pvx{n+1},t_i) &
q_i & \equiv & \introb(g(\vx)\!>\!0, \pvx{n+1},t_i) 
\end{array}
\end{equation}
\noindent Note that if $i \in [0,n]$, then $p_i$ is the same over the
partial signal $\pvx{n}$, \ie, $p_i = \introb(f(\vx)\!>\!0, \pvx{n},
t_i)$ (and respectively for $q_i$). We will use this equivalence in
several of the steps in what follows.

\noindent \textbf{(1)} $\alw \varphi$, where $\varphi \equiv f(\vx) > 0$. Observe the following:
\begin{equation}
\begin{array}{ll}
\introb(\varphi, \pvx{n+1}, 0) 
& = \dmin_{i \in [0,n+1]} p_i 
 = \dmin \left( \dmin_{i \in [0,n]} p_i,\  p_{n+1} \right) 
\end{array}
\end{equation}
In the final expression above, observe that the first entry does not
contain any $p_{n+1}$ terms, \ie, it can be computed using the data
points $\vx_1, \ldots, \vx_n$ in the partial signal $\pvx{n}$ itself.
Thus, for all $n$, if we maintain the one interval representing the
$\dmin$ of the first $n$ values of $f(\vx)$ as a {\em summary}, then
we can compute the interval robustness of $\alw(f(\vx)\!>\!0)$ over
$\pvx{n+1}$ with the additional data $\vx_{n+1}$ available at
$t_{n+1}$. Note for the dual formula $\ev(f(\vx)\!>\!0)$, a similar
result holds with $\dmin$ substituted by $\dmax$.

\noindent \textbf{(2)} $\varphi \Until \psi$, 
where $\varphi \equiv
f(\vx)\!>\!0$, and $\psi \equiv g(\vx)\!>\!0$. 
Observe the following:
\begin{equation}
\label{eq:until_introb}
\introb(\varphi \Until \psi, \pvx{n+1}, 0)  = \dmax_{i \in [0,n+1]} \dmin(q_i, \dmin_{j\in[0,i]} p_j)
\end{equation}
\noindent We can rewrite the RHS of Eq.~\eqref{eq:until_introb} to get:
\begin{equation}
\dmax\left( \underline{\dmax_{i\in[0,n]} \dmin\left(q_i, \dmin_{j\in[0,i]} p_j\right)},\quad
          \dmin\left( \underline{\dmin_{j\in[0,n]} p_j},\ p_{n+1}, q_{n+1} \right) \right) 
\end{equation}
\noindent Let $U_n$ and $M_n$ respectively denote the first and second
underlined terms in the above expression.  Note that for any $n$,
$U_n$ and $M_n$ can be computed only using data $\vx_1,\ldots,\vx_n$.
Consider the recurrences $M_{n+1} = \dmin(M_n, p_{n+1}, q_{n+1})$ and
$U_{n+1} = \dmax(U_n, M_{n+1})$; we can observe that to compute
$M_{n+1}$ and $U_{n+1}$, we only need $M_n$, $U_n$, and $\vx_{n+1}$.
Furthermore, $U_{n+1}$ is the desired interval robustness value over
the partial signal $\pvx{n+1}$.  Thus storing and iteratively updating
the two interval-values $U_n$ and $M_n$ is enough to monitor the given
formula.

\noindent \textbf{(3)} $\alw(\varphi \vee \ev \psi)$, where $\varphi
\equiv f(\vx)\!>\!0$, and $\psi \equiv g(\vx)\!>\!0$.  Observe the
following:
\begin{equation}
\label{eq:alwevone}
\begin{array}{ll}
\introb(\alw(\varphi \vee \ev \psi), \pvx{n+1},0) & 
    = \dmin_{i\in[0,n+1]} \dmax\left(p_i, \dmax_{j\in[i,n+1]} q_j\right)   \\
&   = \dmin_{i\in[0,n+1]} \dmax\left(p_i, \dmax_{j\in[i,n]} q_j, q_{n+1}\right) 
\end{array}
\end{equation}
\noindent Repeatedely applying the equivalence \eqref{eq:equivone}
to the outer $\min$ in \eqref{eq:alwevone} we get:
\begin{equation}
\dmax\left(q_{n+1}, \dmin_{i\in[0,n+1]} \dmax\left(p_i, \dmax_{j\in[i,n]}q_j\right)\right)
\end{equation}
\noindent The inner $\min$ simplifies to:
\begin{equation}
\dmax\left(q_{n+1}, \dmin\left(p_{n+1}, \underline{\dmin_{i\in[0,n]}\left(\dmax\left(p_i, \dmax_{j\in[i,n]} q_j\right)\right)}\right)\right)
\end{equation}
Let $T_n$ denote the underlined term; note that we do not require any
data at time $t_{n+1}$ to compute it. Using the recurrence $T_{n+1} =
\dmax\left(q_{n+1}, \dmin\left(p_{n+1}, T_n\right)\right)$, we can obtain
the desired interval robustness value. The memory required is that for
storing the one interval value $T_n$. A similar result can be established
for the dual formula $\ev(f(\vx)\!>\!0 \wedge \alw (g(\vx)\!>\!0))$.

\noindent \textbf{(4)} $\alw\ev(\varphi)$, where $\varphi \equiv
f(\vx)\!>\!0$. Observe the following:
\begin{equation}
\introb(\alw\ev(\varphi,\pvx{n+1},0) = \dmin_{i \in [0,n+1]} \dmax_{j \in [i,n+1]} p_j 
\end{equation}
Rewriting the outer $\min$ operator and the inner $\max$ more
explicitly, we get:
\begin{equation}
\dmin\left(\underline{\dmin_{i \in [0,n]} \dmax\left(\dmax_{j\in [i,n]} p_j, p_{n+1}\right)}, \quad p_{n+1}\right) 
\end{equation}
Repeatedly using \eqref{eq:equivone} to simplify the above underlined term
we get:
\begin{equation}
\dmin\left(\dmax\left(p_{n+1}, \dmin_{i\in[0,n]} \dmax_{j\in[i,n]}
p_j\right), p_{n+1}\right)  = p_{n+1}.
\end{equation}

The simplification to $p_{n+1}$, follows from \eqref{eq:equivfour}.  Thus,
to monitor $\alw\ev(f(\vx)\!>\!0)$, we do not need to store any
information, as the interval robustness simply evaluates to that of the
predicate $f(\vx)\!>\!0$ at time $t_{n+1}$. A similar result can be
obtained for the dual formula $\ev \alw (f(\vx)\!>\!>0)$.

\noindent \textbf{(5)} $\ev(\varphi \wedge \ev (\psi))$, where
$\varphi \equiv f(\vx)\!>\!0$ $\psi \equiv \ev(g(\vx)\!>\!0))$.
Observe the following:
\begin{equation}
\label{eq:evev}
\introb(\ev( \varphi \wedge \ev(\psi)), \pvx{n+1}, 0) =
\dmax_{i\in[0,n+1]}\left(\dmin\left(p_i, \dmax_{j\in[i,n+1]} q_j\right)\right)
\end{equation}
\noindent We can rewrite the RHS of Eq.~\eqref{eq:evev} as the first
expression below. Applying the equivalence in \eqref{eq:equivtwo} and
\eqref{eq:equivthree} to the expression on the left, we get the expression on
the right.
\begin{equation}
\label{eq:evevone}
\begin{array}{ll}
\dmax\left(
\begin{array}{l}
\min\left( p_0, \max\left(q_0, \ldots, q_{n+1}\right) \right) \\
\cdots\\
\min\left( p_n, \max\left(q_n, q_{n+1}\right) \right) \\
\min\left( p_{n+1}, q_{n+1} \right) \\
\end{array}
\right) & = 
\dmax\left(
\begin{array}{l}
\min(p_0,q_0), \ldots, \min(p_0,q_{n+1}), \\
\cdots \\
\min(p_n,q_n), \min(p_n, q_{n+1}), \\
\min( p_{n+1}, q_{n+1} )
\end{array}
\right)
\end{array}
\end{equation}
Grouping terms containing $q_{n+1}$ together and applying the equivalence
in \eqref{eq:equivtwo} we get:
\begin{equation}
\max\left(
\begin{array}{l}
\max\left(
\begin{array}{l}
\min(p_0,q_0), \min(p_0,q_1), \ldots, \min(p_0,q_n), \\
\min(p_1,q_1), \ldots, \min(p_1,q_n), \\
\cdots \\
\min(p_n,q_n)
\end{array}
\right), \\
\min(q_{n+1}, \underline{\max(p_0,p_1,\ldots,p_n)}),\\
\min(p_{n+1}, q_{n+1})
\end{array}
\right)
\end{equation}
\noindent Observe that the first argument to the outermost $\max$ can be
computed using only $\vx_1,\ldots,\vx_n$. Suppose we denote this term
$T_n$. Also note that in the second argument, the inner $\max$ (underlined)
can be computed using only $\vx_1,\ldots,\vx_n$. Let us denote this term by
$M_n$. We now have a recurrence relations:
\begin{eqnarray}
M_{n+1} & = & \max(M_n, p_{n+1}), \\
T_{n+1} & = & \max(T_n, \min(q_{n+1}, M_n), \min(q_{n+1}, p_{n+1})),
\end{eqnarray}
\noindent where $T_0$ = $\min(p_0,q_0)$ and $M_0$ = $p_0$. Thus, the
desired interval robustness can be computed using only two values stored in
$T_n$ and $M_n$. The dual result holds for the formula 
$\alw(\varphi \vee \alw (\psi))$.
\end{proof}

\noindent{\em Remarks on extending above result:}
The result in Theorem~\ref{thm:finiteness} can be generalized to allow
$\varphi$ and $\psi$ that are not atomic predicates, under following two
conditions:
\begin{enumerate}
\item
Bounded horizon subformulae condition: For each formula, the subformulae
$\varphi$ and $\psi$ have a bounded time-horizon, \ie, $\scope(\varphi)$
and $\scope(\psi)$ are closed intervals.
\item
Smallest step-size condition: Consecutive time-points in the signal
are at least $\Delta$ seconds apart, for some finite $\Delta$, which is
known {\em a priori}.
\end{enumerate}

\subsection{Generalizing Theorem~\ref{thm:finiteness}}

Let $\subf(\varphi)$ denote the set of
all subformulas of $\varphi$ except $\varphi$ itself. Let $\last(\varphi)$
be defined as follows:
\begin{equation}
\label{eq:last}
\last(\varphi) \triangleq \displaystyle \max_{\psi \in \subf(\varphi)} \sup(\scope(\psi))
\end{equation}

The meaning of $\last(\varphi)$ is as follows: the last time at which a
data value of $\vx$ is required to compute $\rob(\varphi,\vx,t)$, is
$t+\last(\varphi)$.  For the formula $\varphi$ defined in
Eq.~\eqref{eq:runf}, $\last(\varphi) = a+c$. For the formula $\psi \equiv
\alw (x>0)$, $\last(\psi) = \infty$.  In general, for any untimed formula
$\varphi$, $\last(\varphi)$ is equal to $\infty$.  In
Theorem~\ref{thm:finiteness}, we show that certain classes of untimed
formulas can be monitored in an online fashion with bounded amount of
memory.  We first define the following quantities:
\begin{equation}
\begin{array}{l@{\hspace{2em}}l@{\hspace{2em}}l}
\Delta \triangleq \displaystyle\min_{i\ge0} (t_{i+1} - t_i) &
w_\varphi \triangleq \displaystyle\max_{\psi\in\subf(\varphi)} \last(\psi) &
k_\varphi \triangleq \left\lceil\frac{w_\varphi}{\Delta}\right\rceil.
\end{array}
\end{equation}

Here, $\Delta$ represents the smallest time-step in the monitored signal,
$w_\varphi$ is the largest time horizon of all subformulas of $\varphi$,
and $k_\varphi$ is the largest number of discrete time-points for the trace
in any $w_\varphi$ interval.

\newcommand{\WL}[1]{\mathsf{W}(#1)}
\newcommand{\WLprime}[1]{\mathsf{W}'(#1)}
\begin{figure*}[!t]
\centering

\begin{tikzpicture}[>=stealth']
\tikzstyle{point}=[circle,fill=black,inner sep=0pt,minimum width=1mm]

\draw[densely dotted,very thin,color=gray,step=0.25cm] (0,1) grid (9,5);
\node[point] (n1) at (0,4.5) {};
\node[point,label=above:$t_{s_i}$] (n3) at (1,4.5) {};
\node[point,label=above:$t_{s_i+1}$] (n4) at (2.5,4.5) {};
\node[point] (n5) at (4,4.5) {};
\node[point] (n6) at (5.5,4.5) {};
\node[point,label=above:$t_i$] (n6) at (7,4.5) {};
\node[point,label=above:$t_{i+1}$] (n7) at (8.5,4.5) {};
\draw[draw=black,arrows={<->}] (2,4.25) -- node[smalltext,below] {$w$} (7,4.25);
\draw[very thin] (2,3.75) -- (2,4.5);
\draw[very thin] (7,3.75) -- (7,4.5);

\node[anchor=west,smalltext] (wl1) at (-1,3.75){Time $t_i$:};
\node[anchor=west,smalltext] (wl2) at (-1,3.5) {$\worklist[\treeNode_\psi]\!= $};
\node[smalltext,anchor=south] (xll) at (1,3.6) {$\Summary$};
\node[smalltext] (wll) at (1,3.5) {$[5,5]$};
\node[smalltext,anchor=250] (xl1) at (2.5,3.6) {$\WL{s_i\!+\!1}$};
\node[smalltext] (wl1) at (2.5,3.5) {$[3,\fhi]$};
\node[smalltext,anchor=south] (xl2) at (4,3.6) {$\WL{s_i\!+\!2}$};
\node[smalltext] (wl2) at (4,3.5) {$[2,\fhi]$};
\node[smalltext,anchor=south] (xl3) at (5.5,3.6) {$\WL{s_i\!+\!3}$};
\node[smalltext] (wl3) at (5.5,3.5) {$[1,\fhi]$};
\node[smalltext,anchor=210] (xl4) at (7,3.6) {$\begin{array}{l}\WL{t_i} =\\ \WL{s_i\!+\!4}\end{array}$};
\node[smalltext] (wl4) at (7,3.5) {$[-1,\fhi]$};

\begin{pgfonlayer}{background}
    \node[rectangle,inner sep=-2pt,draw=black!60!green,fill=green!10,fit=(wll) (wl1) (wl2) (wl3) (wl4)] {};
\end{pgfonlayer}

\draw[draw=black,arrows={<->}] (3.5,3) -- node[smalltext,below] {$w$} (8.5,3);
\draw[semithick,densely dotted] (3.5,2.75) -- (3.5,4.5);
\draw[semithick,densely dotted] (8.5,2.75) -- (8.5,4.5);

\node[anchor=west,smalltext] (wl3) at (-1,2.25){Time $t_{i+1}$:};
\node[anchor=west,smalltext] (wl4) at (-1,2) {$\worklist[\treeNode_\psi]\!= $};
\node[smalltext,draw=red,cross out] (wl1l) at (1,2) {$[5,5]$};
\node[smalltext,anchor=90] (xl1) at (2.5,1.8) {$\Summary$};
\node[smalltext] (wl11) at (2.5,2) {$[3,3]$};
\node[smalltext,anchor=90] (xl2) at (4,1.8) {$\WLprime{s_{i\!+\!1} \!+\!1}$};
\node[smalltext] (wl12) at (4,2) {$[2,\fhi]$};
\node[smalltext,anchor=90] (xl3) at (5.5,1.8) {$\WLprime{s_{i\!+\!1} \!+\!2}$};
\node[smalltext] (wl13) at (5.5,2) {$[2,\fhi]$};
\node[smalltext,anchor=90] (xl4) at (7,1.8) {$\WLprime{s_{i\!+\!1} \!+\!3}$};
\node[smalltext] (wl14) at (7,2) {$[2,\fhi]$};
\node[smalltext,ellipse,draw=black!60!green,thick,densely dotted,inner sep=2pt] (wl15) at (8.5,2) {$[2,\fhi]$};
\node[smalltext,anchor=90] (xl5) at (8.5,1.8) 
    {$\WLprime{t_{i\!+\!1}}$};

\node[smalltext,text=black!60!green] (u1) at (2.5,2.5) {$[3,3]$};
\draw[draw=black!60!green,<->] (2.2,2.5) -- node[smalltext,sloped,text=black!60!green,below] (m) {$\min$} (1,3.3);
\draw[->,draw=black!60!green] (m) to[out=270,in=180] (wl11);
\draw[-latex,draw=black!60!green] (wl1) -- (u1);

\begin{pgfonlayer}{background}
    \node[rectangle,inner sep=-2pt,draw=black!60!green,fill=green!10,fit=(wl11) (wl12) (wl13) (wl14) (wl15)] {};
\end{pgfonlayer}



\end{tikzpicture}

\caption{A depiction of the action of the procedure to update the summary
        while computing $\introb(\alw\psi,\pvx{i},t_0)$. Here, $\WL{j}$ is
        shorthand for $\introb(\psi,\pvx{i},t_j)$ and $\WLprime{j}$ is
        shorthand for $\introb(\psi,\pvx{i+1},t_j)$.
\label{fig:untimed_alw}}
\end{figure*}
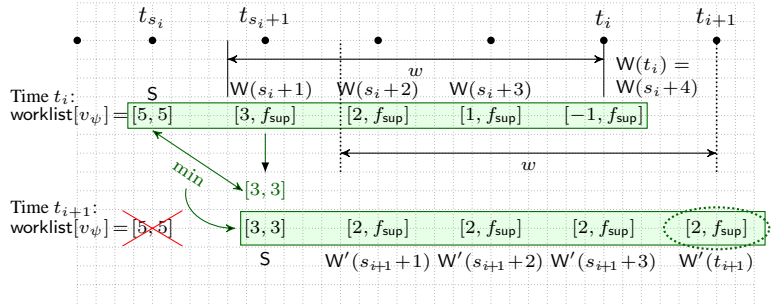


\begin{theorem}
\label{thm:finiteness2}
If $w_\varphi$ is finite, then for each $\varphi$ listed below, we can
monitor \rsi of $\varphi$ in an online fashion using
$O(k_\varphi)$ memory.
\begin{equation}
\begin{array}{l@{\hspace{3em}}l}
1.~\alw \psi \text{ (dually $\ev \psi$)}          & 2.~\psione \Until \psitwo \\
3.~\alw \ev \psi \text{ (dually $\ev \alw \psi$)} & 4.~\alw (\psione \vee \ev \psitwo) \text{(dually $\ev (\psione \wedge \alw \psitwo)$)}, \\
5.~\ev(\psione \wedge \ev \psitwo) \text{ (dually $\alw(\psione \vee \ev \psitwo)$ } &  
\end{array}
\end{equation}
\end{theorem}

\newcommand{\vz}{\mathbf{z}}


\begin{proof} We provide proof sketches. The main argument in each of
the proofs is as follows: For any partial signal $\pvx{i}$, there
are two cases: The first case is when \mbox{$t_0 \ge t_i -
w_\varphi$}. By assumption, there are at most $k_\varphi$
time-points in the interval $[t_0,t_i]$.  Thus, in this case, the
worklists at each of the nodes $\treeNode_\psi$ corresponding to
$\psi \in \subf(\varphi)$ have to track at most $k_\varphi$ \rsi
values in order to compute $\introb(\varphi,\vx,t_0)$.
\vspace{0.5em}

The second case is when \mbox{$t_0 < t_i - w_\varphi$}; this implies
that there is a largest time $t_{s_i}$ in $[t_0,t_1,\ldots,t_i]$ such
that \mbox{$t_{s_i} < t_i - w_\varphi$}.  For the partial signal
$\pvx{i}$, at each time $t \le t_{s_i}$, there is enough information
to compute the exact robustness value of each of the subformulas of
$\varphi$.  The central step is that for each of the formulas
mentioned above, the robustness values in the interval $[t_0,t_{s_i}]$
can be {\em summarized} to a single robustness value.  Furthermore,
the interval $(t_{s_i},t_i]$ can have at most $k_\varphi$ time-points.
Thus, the computation of $\introb(\varphi,\pvx{i},t_0)$ can be split
into tracking a summary for the interval $[t_{s_i},t_i]$ and tracking
at most $k_\varphi$ \rsis in the worklists of the immediate
subformulas of $\varphi$ in the interval $(t_{s_i},t_i]$. We now
explain how the summary information is maintained for each formula.


\vspace{0.5em} \noindent (1) [{\bf $\alw \psi$}]\quad We maintain the
summary $\Summary$ = \mbox{$\inf_{j\in[0,s_i]} \introb(\psi,\pvx{i},t_j)$},
\ie, the infimum over all exact robustness values computable over the
partial signal $\pvx{i}$.  When a new time-point $(t_{i+1},\vx_{i+1})$
becomes available, $\Summary$ is updated if there is a new time
$t_{s_{i+1}}$ for which $\introb(\psi,\pvx{i},t_{s_{i+1}})$ can be exactly
computed; otherwise, the new value is used to update all entries
$\introb(\psi,\pvx{i+1},t_j)$ for $t_j \in [t_{s_i + 1},t_i]$, and a new
entry corresponding to time $t_{i+1}$ is added to
$\worklist[\treeNode_\psi]$. Please see Fig.~\ref{fig:untimed_alw} for a
depiction of this step. We then establish the following: (1) There are at
most $k_\varphi$ entries (each corresponding to $\introb(\psi,\pvx{i},t_j)$
for $t_j \in (t_{s_i}, t_i]$) in $\worklist[\treeNode_\psi]$. This is true
because there can be at most $k_\varphi$ consecutive time-points that do
not update $\Summary$ in any interval of length $w_\varphi$.   (2) We show
by induction that the $\inf$ of $\Summary$ and the $k_\varphi$ entries in
$\worklist[\treeNode_\psi]$ is equal to $\introb(\alw\psi,\pvx{i},t_0)$.
\vspace{0.5em}


\noindent (2) [{\bf $\psione \Until \psitwo$}]\quad We maintain the
following two quantities as the summary:\\ (a) $\Sone$ =
$\introb(\alw_{[0,t_{s_i}]}\psione, \pvx{i}, t_0)$ and (b) $\Stwo$ =
\mbox{$\introb(\psione \Until_{[0,t_{s_i}]} \psitwo, \pvx{i}, t_0)$}.  In
$\worklist[\treeNode_\psione]$ and $\worklist[\treeNode_\psitwo]$ we store
at most $k_\varphi$ values corresponding to $\introb(\psione, \pvx{i},
t_j)$ and $\introb(\psitwo, \pvx{i}, t_j)$ for $t_j \in (t_{s_i},t_i]$.
The crucial step is to combine $\Sone$ and $\Stwo$ with the entries in
$\worklist[\treeNode_\psione]$ and $\worklist[\treeNode_\psitwo]$ to obtain
$\introb(\psione \Until \psitwo, \pvx{i},t_0)$.  We show that the iterative
procedure in Algorithm~\ref{algo:compute_untimed_until} can accomplish
this.  In its $j^{th}$ iteration $v_1$ is equal to $\dinf_{\ell \in [0,j]}
\introb(\psione,\pvx{i},t_\ell)$, and we can show by induction that $v_2$
is equal to $\dsup_{m \in [0,j]} \min\left(\introb(\psitwo, \pvx{i}, t_m),
\dinf_{\ell \in [0,m]} \introb(\psione,\pvx{i}, t_\ell)\right)$.  Thus, at
the end of the computation, the value computed in $v_2$ is
$\introb(\psione\Until\psitwo, \pvx{i}, t_0)$. 
\begin{algorithm}[t!]
\DontPrintSemicolon
$v_1$ $\assign$ $\Sone$, $v_2$ $\assign$ $\Stwo$ \;
\ForEach{$j \in [s_i+1,i]$}{
    $v_1$ $\assign$ $\min\left(v_1, \introb(\psione, \pvx{i}, t_j) \right)$ \;
    $v_2$ $\assign$ $\sup\left(v_2, \min\left(v_1, \introb(\psitwo,\pvx{i}, t_j)\right)\right)$ \;
}
$\introb(\psione\Until\psitwo, \pvx{i}, t_0)$ $\assign$ $v_2$ \;
\caption{Computing \rsi for untimed Until}
\label{algo:compute_untimed_until}
\end{algorithm}

\vspace{0.5em}
\noindent (3) [{\bf $\alw\ev \psi$}]\quad
We show that we do not need any additional storage for monitoring
$\varphi$.  Concretely, we posit that $\introb(\alw\ev\psi, \pvx{i}, t_0)$
= $\introb(\psi,\pvx{i},t_i)$. We successively rewrite
\mbox{$\introb(\alw\ev\psi, \pvx{i}, t_0)$} = $\dinf_{j\in[0,i]} \sup_{\ell
\in[j,i]} \introb(\psi, \pvx{i}, t_\ell)$ as follows:
\vspace{-0.4em}
\begin{eqnarray}
\inf\left(
    \dsup_{\ell \in[0,i]} \introb(\psi,\pvx{i},t_\ell),
    \dsup_{\ell \in[1,i]} \introb(\psi,\pvx{i},t_\ell),
    \ldots,
    \dsup_{\ell \in [i,i]} \introb(\psi,\pvx{i},t_\ell) 
\right) \label{eq:expand_one} \\
\inf\left(\begin{array}{l}
    \sup(\introb(\psi,\pvx{i},t_i),\quad \sup_{\ell \in[0,i-1]} \introb(\psi,\pvx{i},t_\ell)) \\
    \sup(\introb(\psi,\pvx{i},t_i),\quad \sup_{\ell \in[1,i-1]}
    \introb(\psi,\pvx{i},t_\ell)), \ldots, \introb(\psi,\pvx{i},t_\ell) 
\end{array} \right) \label{eq:expand_two}
\end{eqnarray}
In the above, to go from \eqref{eq:expand_one} to \eqref{eq:expand_two}, we
expand the inner $\sup$ expressions, and observe that the last term in the
$\inf$ evaluates to $\introb(\psi,\pvx{i},t_i)$. For the final step, we
observe that $\inf(I_1, \sup(I_1,I_2), \ldots, \sup(I_1,I_n)) = I_1$, and
thus, \eqref{eq:expand_two} simplifies to $\introb(\psi,\pvx{i},t_i)$ By
duality, a similar proof works for $\ev\alw\psi$. 

\vspace{0.5em}
\noindent (4) [{\bf $\alw(\psione \vee \ev \psitwo)$}] \quad
We maintain one quantity as the summary information: $\Summary$ =
$\introb(\alw_{[0,s_i]}(\psione \vee \ev \psitwo), \pvx{i}, t_0)$.
Additionally, we store at most $k_\varphi$ entries corresponding to
$\introb(\psione,\pvx{i},t_j)$ in $\worklist[\treeNode_\psione]$ and at
most $k_\varphi$ entries corresponding to $\introb(\psitwo,\pvx{i},t_j)$ in
$\worklist[\treeNode_\psitwo]$.  To compute $\introb(\varphi,\pvx{i},t_0)$,
we use Algorithm~\ref{algo:compute_alw_ev}.
\begin{algorithm}[t!]
\DontPrintSemicolon
$v$ $\assign$ $\Summary$ \;
\ForEach{$j \in [s_i+1,i]$}{
    $v$ $\assign$ $\sup(\introb(\psitwo,\pvx{i},t_j),\inf(v,\introb(\psione,\pvx{i},t_j)))$ \;
}
$\introb(\alw(\psione \vee \ev\psitwo), \pvx{i}, t_0)$ $\assign$ $v$ \;
\caption{Computing \rsi for $\alw(\psione \vee \ev\psitwo)$}
\label{algo:compute_alw_ev}
\end{algorithm}
To complete the proof we observe that Algorithm~\ref{algo:compute_alw_ev}
computes expression \eqref{eq:alw_ev} that has nested and alternating
$\sup$s and $\inf$s:
\begin{equation}
\label{eq:alw_ev}
\sup\left(\introb(\psitwo,\pvx{i},t_i),
\inf\left(\introb(\psione,\pvx{i},t_i), 
\sup\left(\introb(\psitwo, \pvx{i},t_{i-1}), \ldots, 
\Summary\right)\right)\right)
\end{equation}
Using the identity $\sup(I_1, \inf(I_2, I_3))$ = $\inf(\sup(I_1,I_2),
\sup(I_1,I_3))$, we can rearrange the above expression to obtain:
\begin{equation}
\inf\left(\begin{array}{l}
    \sup\left(\introb(\psitwo,\pvx{i},t_i),\introb(\psione,\pvx{i},t_i)\right), \\
    \sup\left(\begin{array}{l}
              \introb(\psitwo,\pvx{i},t_i),\introb(\psitwo,\pvx{i},t_{i-1}), \\
              \inf\left(\introb(\psione,\pvx{i},t_{i-1}), 
                        \sup\left(\introb(\psitwo,\pvx{i},t_{i-2}), \ldots, \Summary\right)
                   \right)
              \end{array}
            \right)
    \end{array}\right)
\end{equation}
By repeated use of this identity on the expression in the second line, we
get the expression $\dinf_{j \in [0,i]} \left(\max\left(\introb(\psione,\pvx{i},t_j),
\dsup_{\ell \in [j,i]} \introb(\psitwo, \pvx{i}, t_\ell)\right)\right)$, which is
equal to $\introb(\varphi,\pvx{i},t_0)$. \hfill $\blacksquare$
\end{proof}


\section{Experimental Results}\label{sec:experiments}
We implemented Algorithm~\ref{algo:online_monitoring} as a stand-alone tool
that can be plugged in loop with any black-box simulator and evaluated it
using two practical real-world applications. We considered the following
criteria: (1) On an average, what fraction of simulation time can be saved
by online monitoring? (2) How much overhead does online monitoring add, and
how does it compare to a na\"{i}ve implementation that at each step
recomputes everything using an offline algorithm?

\subsection{Diesel Engine Model (DEM)}
The first case study is an industrial-sized Simulink\rt model of a
prototype airpath system in a diesel engine. The closed-loop model consists
of a plant model describing the airpath dynamics, and a controller
implementing a proprietary control scheme.  The model has more than $3000$
blocks, with more than $20$ lookup tables approximating high-dimensional
nonlinear functions.  Due to the significant model complexity, the speed of
simulation is about $5$ times slower, \ie, simulating $1$ second of
operation takes $5$ seconds in Simulink\rt.  As it is important to simulate
this model over a long time-horizon to characterize the airpath behavior
over extended periods of time, savings in simulation-time by early
detection of requirement violations is very beneficial.  We selected two
parameterized safety requirements  after discussions with the control
designers, (shown in Eq.~\eqref{eq:diesel_prop1}-\eqref{eq:diesel_prop2}).
Due to proprietary concerns, we suppress the actual values of the
parameters used in the requirements.
\begin{eqnarray}
\fover(\mathbf{p_1}) & = & \alw_{[a,b]}(\vx < c) 
    \label{eq:diesel_prop1} \\
\ftemp(\mathbf{p_2}) & = & \alw_{[a,b]}(\abs{\vx} > c\! \implies\!\!(\ev_{[0, d]}\!\abs{\vx} < e ))
    \label{eq:diesel_prop2}
\end{eqnarray}

Property $\fover$ with parameters $\mathbf{p_1} = (a,b,c)$ specifies
that in the interval $[a,b]$, the overshoot on the signal $\vx$ should
remain below a certain threshold $c$.  Property $\ftemp$ with
parameters $\mathbf{p_2} = (a,b,c,d,e)$ is a specification on the
settling time of the signal $\vx$.  It specifies that in the time
interval $[a,b]$ if at some time $t$, $\abs{\vx}$ exceeds $c$ then it
settles to a small region ($\abs{\vx} < e$) before $t+d$.  In
\tabref{diesel}, we consider three different valuations $\nu_1$,
$\nu_2$, $\nu_3$ for $\mathbf{p_1}$ in the requirement
$\fover(\mathbf{p_1})$, and two different valuations $\nu_4$, $\nu_5$
for $\mathbf{p_2}$ in the requirement $\ftemp(\mathbf{p_2})$.
\vspace{0.5em}

The main reason for the better performance of the online algorithm is that
simulations are time-consuming for this model.  The online algorithm can
terminate a simulation earlier (either because it detected a violation or
obtained a concrete robust satisfaction interval), thus obtaining
significant savings.  For $\fover(\nu_3)$, we choose the parameter values
for $a$ and $b$ such that the online algorithm has to process the entire
signal trace, and is thus unable to terminate earlier.  Here we see that
the total overhead (in terms of runtime) incurred by the extra book-keeping
by Algorithm~\ref{algo:online_monitoring} is negligible (about $0.1$\%).

\begin{table}[t]
\centering
\begin{tabular*}{.99\textwidth}{@{\extracolsep{\fill}}l cc cc}
\toprule
Requirement     & Num.   & Early         & \mtc{Simulation Time (hours)} \\
\cline{4-5}
                & Traces & Termination   & Offline  & Online        \\
\midrule          
\vspace{-0.2em}
$\fover(\nu_1)$ & 1000   & 801           & 33.3803  & 26.1643      \\
\vspace{-0.2em}
$\fover(\nu_2)$ & 1000   & 239           & 33.3805  & 30.5923      \\
\vspace{-0.2em}
$\fover(\nu_3)$ & 1000   & 0             & 33.3808  & 33.4369      \\
\vspace{-0.2em}
$\ftemp(\nu_4)$ & 1000   & 595           & 33.3822  & 27.0405      \\
\vspace{-0.2em}
$\ftemp(\nu_5)$ & 1000   & 417           & 33.3823  & 30.6134      \\
\bottomrule
\end{tabular*}
\caption{Experimental results on DEM.}
\label{tab:diesel}
\end{table}

\subsection{CPSGrader}
CPSGrader~\cite{juniwal-emsoft14,cpsgrader-www} is a publicly-available
automatic grading and feedback generation tool for online virtual labs in
cyber-physical systems.  It employs temporal logic based testers to check
for common fault patterns in student solutions for lab assignments.
CPSGrader uses the National Instruments Robotics Environment Simulator to
generate traces from student solutions and monitors STL properties (each
corresponding to a particular faulty behavior) on them. In the published
version of CPSGrader~\cite{juniwal-emsoft14}, this is done in an offline
fashion by first running the complete simulation until a pre-defined
cut-off and then monitoring the STL properties on offline traces. At a
step-size of $5$ ms, simulating $6$ sec. of real-world operation of the
system takes $1$ sec. for the simulator.  When students use CPSGrader for
active feedback generation and debugging, simulation constitutes the major
chunk of the application response time. Online monitoring helps in reducing
the response time by avoiding unnecessary simulations, giving the students
feedback as soon as faulty behavior is detected.

We evaluated our online monitoring algorithm, on the traces and STL
properties used in the published version of
CPSGrader~\cite{juniwal-emsoft14,cpsgrader-www}. These traces are the
result of running actual student submissions on a battery of tests. For
lack of space, we refer the reader to~\cite{juniwal-emsoft14} for details
about the tests and STL properties. As an illustrative example, we show the
$\tbKB$ property in Eq.~\ref{eq:keep_bump}: 
\begin{equation} \label{eq:keep_bump}
\fkb = \ev_{[0, 60]} \alw_{[0, 5]}\left(\bumpr(t) \lor \bumpl(t)\right)
\end{equation}
For each STL property, Table~\ref{tab:mooc_study} compares the total
simulation time needed for both the online and offline approaches, summed
over all traces.  For the offline approach, a suitable simulation cut-off
time of $60$ sec. is chosen.  At a step-size of 5 ms, each trace is roughly
of length $1000$. For the online algorithm, simulation terminates before
this cut-off if the truth value of the property becomes known, otherwise it
terminates at the cut-off. Table~\ref{tab:mooc_study} also shows the
monitoring overhead incurred by a na\"{i}ve online algorithm that performs
complete recomputation at every step against the overhead incurred by
Alg.~\ref{algo:online_monitoring}.
\begin{table}[t]
    \setlength{\tabcolsep}{2pt}
\centering
\footnotesize{%
        \begin{tabular*}{.99\textwidth}{@{\extracolsep{\fill}}l cc rr rr}
\toprule
STL Test Bench &  Num.    & Early         & \mtc{Sim. Time (mins)} & \mtc{Overhead (secs)}\\
               &  Traces  &  Termination  & Offline & Online       & Na\"{i}ve & Alg.~\ref{algo:online_monitoring} \\
\midrule
\vspace{-0.2em}
$\tbAF$ & 1776 & 466 & 296 & 258 &  553 &  9 \\
\vspace{-0.2em}
$\tbAL$ & 1778 & 471 & 296 & 246 & 1347 & 30 \\
\vspace{-0.2em}
$\tbAR$ & 1778 & 583 & 296 & 226 & 1355 & 30 \\
\vspace{-0.2em}
$\tbHCa$& 1777 &  19 & 395 & 394 &  919 & 11 \\
\vspace{-0.2em}
$\tbHCb$& 1556 & 176 & 259 & 238 &  423 &  7 \\
\vspace{-0.2em}
$\tbHCc$& 1556 & 124 & 259 & 248 &  397 &  7 \\
\vspace{-0.2em}
$\tbF$  & 1451 &  78 & 242 & 236 &  336 &  6 \\
\vspace{-0.2em}
$\tbKB$ & 1775 & 468 & 296 & 240 & 1.2$\times10^4$ & 268 \\
\vspace{-0.2em}
$\tbWH$ & 1556 &  71 & 259 & 253 & 1.9$\times10^4$ & 1.5$\times10^3$ \\
\bottomrule
\end{tabular*}
}

\caption{\small{Evaluation of online monitoring for CPSGrader. Each STL Test
        Bench has an associated STL property. 
}}

\label{tab:mooc_study}
\end{table}
Table~\ref{tab:mooc_study} demonstrates that online monitoring ends up
saving up to 24\% simulation time ($>10$\% in a majority of cases). The
monitoring overhead of Alg.~\ref{algo:online_monitoring} is negligible
($<1$\%) as compared to the simulation time and it is less than the
overhead of the na\"{i}ve online approach consistently by a factor of 40x
to 80x.



\bibliographystyle{abbrv}
\bibliography{monitoring}

\end{document}